\documentclass[conference]{IEEEtran}
\IEEEoverridecommandlockouts

\usepackage{cite}
\usepackage{amsmath,amssymb,amsfonts}
\usepackage{algorithmic}
\usepackage{graphicx}
\usepackage{textcomp}
\usepackage{xcolor,soul}

\usepackage{amsthm,amsmath,amssymb}

\usepackage{bm}
\usepackage{mathrsfs}
\usepackage{xcolor}

\usepackage{tikz,xcolor,graphicx}

\usepackage{mathtools}

\usepackage{empheq}

\usepackage{xargs}

\usepackage[unicode]{hyperref}
\usepackage{url}
\usepackage{hyperref,verbatim}

\usepackage{makecell}

\usepackage[textsize=small,textwidth=.8in]{todonotes}
\reversemarginpar

\definecolor{JLHgreen}{RGB}{50,100,50}

\newcommand{\mh}[1]{\textcolor{red}{#1}}

\newcommand{\david}[1]{\textcolor{blue}{#1}}

\newcommandx{\now}[2][1=]{\todo[inline, linecolor=red,backgroundcolor=red!25,bordercolor=red,#1]{#2}}

\newcommandx{\add}[2][1=]{\todo[inline, linecolor=teal,backgroundcolor=teal!25,bordercolor=teal,#1]{#2}}
\newcommandx{\improve}[2][1=]{\todo[inline, linecolor=violet,backgroundcolor=violet!25,bordercolor=violet,#1]{#2}}
\newcommandx{\change}[2][1=]{\todo[inline, linecolor=blue,backgroundcolor=blue!25,bordercolor=blue,#1]{#2}}
\newcommandx{\delete}[2][1=]{\todo[inline, linecolor=orange,backgroundcolor=orange!25,bordercolor=orange,#1]{#2}}

\usepackage{float}

\usepackage{cancel}
\usepackage{ulem}

\renewcommand{\emph}[1]{#1}

\usepackage{environ}
\NewEnviron{eqs}{%
  \begin{equation}\begin{split}
    \BODY
  \end{split}\end{equation}
}

\usepackage{color,xcolor}

\newtheorem{theorem}{Theorem}

\newtheorem{lemma}[theorem]{Lemma}

\newtheorem{corollary}[theorem]{Corollary}

\newtheorem{definition}[theorem]{Definition}
\newtheorem{example}[theorem]{Example}
\newtheorem{remark}[theorem]{Remark}

\newcommand{\FF}{\mathbb{F}}
\newcommand{\NN}{\mathbb{N}}
\DeclareMathOperator{\Ker}{ker}
\DeclareMathOperator{\Ima}{im}

\def \a{\alpha}
\def \b{\beta}
\def \v{\nu}

\begin{document}

\title{$c^3$-Locally Testable Codes from Lossless Expanders}

\author{
\IEEEauthorblockN{Ting-Chun Lin}
\IEEEauthorblockA{
    \textit{University of California San Diego} \\
    \textit{and Hon Hai (Foxconn) Research Institute} \\
    Taipei, Taiwan \\
    til022@ucsd.edu
}
\and
\IEEEauthorblockN{Min-Hsiu Hsieh}
\IEEEauthorblockA{
    \textit{Hon Hai (Foxconn) Research Institute} \\
    Taipei, Taiwan \\
    min-hsiu.hsieh@foxconn.com
}
}

\sloppy

\maketitle

\begin{abstract}

A locally testable code (LTC) is an error correcting code with a property tester. The tester tests if a word is codeword by reading constant random bits and rejects the word with probability proportional to the distance from the word to the closest codeword. An important open question until recently is whether there exist $c^3$-LTCs which are LTCs with constant rate, constant relative distance and constant locality. In this work, we construct a new LTC family using 1-sided lossless expanders and balanced products. 
\end{abstract}


\section{Introduction}
\label{sec:introduction}

The study of error correcting code has a long history dating back to Shannon's work \cite{shannon1948mathematical}.
He showed that there exist codes with constant rate and linear distance by taking a random codebook.
This simple idea springs to many different generalizations in different settings.
In the case of locally testable codes (LTC), one asks if there exist codes with constant rate, linear distance and a local tester.
The local tester tests if a word is a codeword
  by reading $w$ bits of the word
  and rejects the word with probability proportional to its distance from the closest codeword.
When the locality parameter $w$ is a constant, the code is called the $c^3$-LTCs, which stands for constant rate, constant relative distance, and constant locality.
Similar variations that concern locality include low-density parity-check codes (LDPC) \cite{gallager1962low}, locally decodable codes (LDC), locally correctable codes (LCC) \cite{kopparty2017high}.

For LDPC codes, a random sparse construction gives constant rate and linear distance.
For LTC, however, a random construction does not guarantee the existence of a tester.
Nevertheless, there are techniques beyond random construction and better LTC has been constructed over time \cite{polishchuk1994nearly, goldreich2006locally, ben2003randomness, ben2006robust, ben2005simple, dinur2007pcp, kopparty2017high, gopi2018locally}.
Many of those works have close connections with probabilistically
checkable proofs (PCP) \cite{arora1998proof}
which also has a local tester.
Despite these progress, the existence of the ultimate $c^3$-LTCs remains uncertain. 
Recently, two papers \cite{panteleev2021asymptotically} \cite{dinur2021locally} show the existence of $c^3$-LTCs using Tanner codes.
In this work, we show the existence of $c^3$-LTCs using lossless expanders.

\subsection{Main results and technical tools}

In this paper, we construct an explicit family of LTCs with rate arbitrarily close to $1$. 

The construction is based on the balanced product of two 1-sided lossless expander graphs with group symmetry.
Balanced product \cite{breuckmann2021balanced}
  is the Cartesian product of 
  two graphs with a common free group action
  quotient over the diagonal group action.
A 1-sided lossless expander \cite{capalbo2002randomness}
  is a regular bipartite graph with vertex expansion approximately equal to its degree.
The balanced product of two bipartite graphs produces a new graph with two-dimensional features.
We then take the adjacency matrix of the new graph as the parity-check matrix to obtain the code.

The most difficult part in proving the code being LTC is to show constant soundness. We do so by showing another code from the graph has a new property called the small set LTC.
  
\subsection{Outline}
Section~\ref{sec:preliminary} reviews the definition of LTCs and the mathematical tools used in the construction.
Section~\ref{sec:distance-LTC} constructs and proves the existence of $c^3$-LTC.

\section{Preliminary} 
\label{sec:preliminary}
\subsection{Classical error correcting codes}

Here, we review classical linear code and locally testable code. 


A classical linear code $C=C(H)$ is described through a parity-check matrix {$H \in \FF_2^{m\times n}$}, where $C(H) = \{ \bm{x} \in \FF_2^n : H \bm{x} = 0 \}$. 
A vector in $C$ is called a \emph{codeword}. 
{The (Hamming) weight of a vector $\bm{x} \in \FF_2^n$ is the number of non-zero entries.} Important parameters of the code $C(H)$ are
\begin{itemize}
\item the \emph{length} of the code: $n$;
\item the \emph{dimension} of the code: $k = \dim(C)$;
\item the \emph{distance} of the code: $d$, which is the minimum Hamming weight of a non-zero codeword;
\item the weight of the code: $w$, which is the maximal weight of all column and row vectors in $H$.
\end{itemize}

A classical linear code $C_{\rm gen}(M)$ can also be defined through a generator matrix $M \in \FF_2^{m' \times n}$ where $C_{\rm gen}(M)$ is the row span of $M$.

Now, we define classical locally testable codes (LTCs). We consider the definition of a ``strong" LTC which implies all other definitions of locally testable codes
  \cite{goldreich2006locally}
  \cite{goldreich2017introduction}.

\begin{definition}[locally testable code (LTC)]
  A classical linear code $C=C(H)$ is $(w,s)$-\emph{locally testable}
    if it has a parity-check matrix $H \in \FF_2^{m\times n}$
    such that the weight of row vectors are at most $w$
    and for any vector $\bm{x} \in \FF_2^n$
  \begin{equation}
    \frac{1}{m} |H \bm{x}| \ge \frac{s}{n} d(\bm{x}, C),
  \end{equation}
  where $d(\bm{x}, C) = \min_{c \in C} |\bm{x} - c|$
    and $|\cdot|$ is the Hamming weight.
\end{definition}

Note that $w$ is related to the number of queries and $s$ is related to soundness in the original definition of LTC.

\subsection{Graphs}

Here, we review the definition of bipartite graphs, adjacency matrices, regular graphs, Cayley graphs, and lossless expander graphs.

We use $\Xi \equiv (V_0, V_1, E)$ to denote a bipartite graph, where $V_0$, $V_1$ are the sets of vertices on the 2 sides, and $E \subseteq V_0 \times V_1$ is the set of edges between the vertices. We use variables $\nu_0, \nu_1$ to denote subsets of $V_0, V_1$, and use variables $x_0, x_1$ to denote individual vertices in $V_0, V_1$. We use $E(\nu_0, \nu_1)$ to denote the set of edges between $\nu_0$ and $\nu_1$.
The neighbors of a vertex $x_0$ within $\nu_1$ is denoted as $N_{\nu_1}(x_0)$. 
The degree of a vertex $x_0$ in $\nu_1$ is defined as the size of the neighbors of $x_0$ in $\nu_1$, i.e. $\deg_{\nu_1}(x_0) \coloneqq |N_{\nu_1}(x_0)|$. 

The adjacency matrix of a bipartite graph 
  is a matrix $L(E) \in \FF_2^{V_0 \times V_1}$,
  where $L(E)_{x_0, x_1} = 1$ if $(x_0, x_1) \in E$,
  otherwise $L(E)_{x_0, x_1} = 0$.
Equivalently,
  $L(E) e_{x_0} = \sum_{x_1 \in N_{V_1}(x_0)} e_{x_1}$,
  where $e_{x_0}$ and $e_{x_1}$ are the basis vectors in $\FF_2^{V_0}$ and $\FF_2^{V_1}$.

A bipartite graph $\Xi$ is $(w_0, w_1)$-regular if the degrees of all vertices in $V_0$ are equal to $w_0$, and the degrees of the all vertices in $V_1$ are equal to $w_1$. Notice that $w_0$ and $w_1$ are the weights of the column and row vectors of the adjacency matrix.

\subsubsection{Cayley graphs and graphs with free group action}

Here, we discuss graphs with free group action, which is crucial for the balanced product construction. Cayley graph is the key example for a graph with free group action. In fact, all graphs with free group action can be decomposed into Cayley graphs.

A bipartite graph $\Xi$ is $G$-invariant if there exist $G$-actions on $V_0$ and $V_1$, such that if $(x_0, x_1) \in E$, then $(gx_0, gx_1) \in E$.
We also say the graph has $G$-symmetry.
Later we only consider a special case of $G$-action, where the action is free.
A group action on a set $V$ is free, if for all $x \in V$, $gx = x$ implies $g = 1$.
A group action on $\Xi$ is free, if the actions on both $V_0$ and $V_1$ are free.



The left (acting) bipartite Cayley graph,
  $\Gamma_{\textnormal{left}}(G, A) = (V_0, V_1, E)$, 
  is a bipartite graph constructed from a group $G$ 
  and a generating set $A \subseteq G$. 
The graph has vertices $V_0 = G, V_1 = G$, 
  and edges $E = \{(g, ag) : g \in G, a \in A\}$. 
We can also have the generating set acts from the right, 
  which defines the right bipartite Cayley graph, 
  $\Gamma_{\textnormal{right}}(G, B)$. 
The left (right) bipartite Cayley graph is $G$-invariant by the right (left) group action which acts freely.

\subsubsection{Lossless expander}

A lossless expander graph is a regular graph where the vertex expansion is approximately equal to its degree. 

\begin{definition}[Small set vertex expansion]
A bipartite graph, $\Xi \equiv (V_0, V_1, E)$, has $(c, \alpha)$-vertex expansion from $V_0$ to $V_1$ if for any subset $\nu_0 \subseteq V_0$ with $|\nu_0| < c |V_0|$, $|N_{V_1}(\nu_0)| \ge \alpha |\nu_0|$. 
\end{definition}



\begin{definition}[1-sided lossless expander]
A $(w_0, w_1)$-regular bipartite graph $\Xi $ is a 1-sided $(c, \epsilon)$-lossless expander from $V_0$ to $V_1$, if it has $(c, (1-\epsilon)w_0)$-vertex expansion from $V_0$ to $V_1$.
\end{definition}

\subsection{Balanced product construction}


The balanced product construction is a general method that can be applied to any two objects with a common group action. 
In our case, it would be two graphs with a common free group action.
The balanced product is obtained by first taking the Cartesian product, then taking the quotient over the diagonal group action.



We first review the definition of hypergraph product \cite{tillich2013quantum}
  which is the same as Cartesian product. 


\begin{definition}[Hypergraph product]
Given two bipartite graphs $\Xi_X = (V_{X, 0}, V_{X, 1},E_X)$ and $\Xi_Y = (V_{Y, 0}, V_{Y, 1}, E_Y)$, the \emph{hypergraph product} of $\Xi_X$ and $\Xi_Y$, $\Xi_X \times \Xi_Y$, has 
\begin{itemize}
  \item vertices: $V_{00} = V_{X,0} \times V_{Y,0},
    V_{10} = V_{X,1} \times V_{Y,0},
    V_{01} = V_{X,0} \times V_{Y,1},
    V_{11} = V_{X,1} \times V_{Y,1}$,
  \item edges: \\
    $E_{*0} = \{((x_0,y_0), (x_1,y_0)) : (x_0,x_1) \in E_X, y_0 \in V_{Y,0}\}, \\
    E_{*1} = \{((x_0,y_1), (x_1,y_1)) : (x_0,x_1) \in E_X, y_1 \in V_{Y,1}\}, \\
    E_{0*} = \{((x_0,y_0), (x_0,y_1)) : x_0 \in V_{X,0}, (y_0,y_1) \in E_Y\}, \\
    E_{1*} = \{((x_1,y_0), (x_1,y_1)) : x_1 \in V_{X,1}, (y_0,y_1) \in E_Y\}$,
  \item faces: $F = \{((x_0,y_0), (x_1,y_0), (x_0,y_1), (x_1,y_1)) : (x_0,x_1) \in E_X, (y_0,y_1) \in E_Y\}$.
\end{itemize}
\end{definition}  

\begin{figure}
  \centering
  \begin{tikzpicture}
    \node (V00) at (0,0) {$V_{00}$};
\node (V10) at (0,-2) {$V_{10}$};
\node (V01) at (2,0) {$V_{01}$};
\node (V11) at (2,-2) {$V_{11}$};

\draw (V00) to node[auto, swap] {$E_{*0}$} (V10);
\draw (V00) to node[auto] {$E_{0*}$} (V01);
\draw (V01) to node[auto] {$E_{*1}$} (V11);
\draw (V10) to node[auto, swap] {$E_{1*}$} (V11);

\node (F) at (1,-1) {$F$};
  \end{tikzpicture}
  \caption{Balanced product of bipartite graphs.}
  \label{fig:hypergraph_BG}
\end{figure}
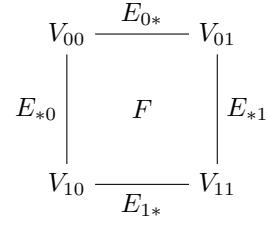


When the graph $\Xi_X$ and $\Xi_Y$ have $G$-action, the hypergraph product $\Xi_X \times \Xi_Y$ has a $G$-action defined by the diagonal $G$-action. After quotienting this action, we obtain the balanced product.

\begin{definition}[Balanced product]
The \emph{balanced product} of two bipartite graphs $\Xi_X, \Xi_Y$ with free $G$-action, denoted by $\Xi_X \times_G \Xi_Y$, 
  has 
\begin{itemize}
  \item vertices: $V_{\a\b} \coloneqq V_{X,\a} \times V_{Y,\b} / G,$
    for $\a, \b \in \{0, 1\}$,
    where $(x_\a, y_\b) \sim (gx_\a, gy_\b)$,
  \item edges: for $\a, \b \in \{0, 1\}$, \\
    \hbox to\columnwidth{%
        $E_{*\b} \coloneqq \{((x_0,y_\b), (x_1,y_\b)) : (x_0,x_1) \in E_X, y_\b \in V_{Y,\b}\} / G$,
    }\\
    where  $((x_0,y_\b), (x_1,y_\b)) \sim ((gx_0,gy_\b), (gx_1,gy_\b))$, and \\
    \hbox to\columnwidth{%
      $E_{\a*} \coloneqq \{((x_\a,y_0), (x_\a,y_1)) : x_\a \in V_{X,\a}, (y_0,y_1) \in E_Y\} / G$,
    }\\
    where $((x_\a,y_0), (x_\a,y_1)) \sim ((gx_\a,gy_0), (gx_\a,gy_1))$,
  \item faces:
     $F \coloneqq \{((x_0,y_0), (x_1,y_0), (x_0,y_1), (x_1,y_1)) :$ \\
    \hbox{}\hfill $(x_0,x_1) \in E_X, (y_0,y_1) \in E_Y\} / G$, \\
    where $((x_0,y_0), (x_1,y_0), (x_0,y_1), (x_1,y_1)) \sim$ \\
    \hbox{}\hfill $((gx_0,gy_0), (gx_1,gy_0), (gx_0,gy_1), (gx_1,gy_1))$.
\end{itemize}
\end{definition}

The balanced product of bipartite graphs is illustrated in 
  Figure~\ref{fig:hypergraph_BG}.
We use $(V^*, E^*, F)$ as a shorthand for
  $(V_{00}, V_{10}, V_{01}, V_{11}, E_{*0}, E_{*1}, E_{0*}, E_{1*}, F)$.


An important example of a balanced product graph is the left-right Cayley graph. 

\begin{definition}[Left-right Cayley graph] \label{def:left-right-cayley-graph}

The left-right bipartite Cayley graph $\Gamma_2(G, A, B) \coloneqq \Xi_X \times_G \Xi_Y$
  where $\Xi_X \coloneqq\Gamma_{\rm right}(G, A^{-1}), \Xi_Y\coloneqq\Gamma_{\rm right}(G, B)$
  are right bipartite Cayley graphs with $A^{-1} = \{a^{-1} : a \in A\}$.

Explicitly, the graph has
\begin{itemize}
  \item vertices: $V_{00} \cong V_{10} \cong V_{01} \cong V_{11} \cong G \times G / G \cong G$,
  \item edges: 
  \begin{align*}
    E_{*0} &= \{(g, ag) : g \in G, a \in A\}, \\
    E_{*1} &= \{(gb, agb) : gb \in G, a \in A\}, \\
    E_{0*} &= \{(g, gb) : g \in G, b \in B\}, \\
    E_{1*} &= \{(ag, agb) : ag \in G, b \in B\},
  \end{align*}
  \item faces: $\{(g, ag, gb, agb): g \in G, a \in A, b \in B\}$,
\end{itemize}
where $G \times G / G \cong G$
uses the bijection $[(x_\a, y_\b)] \mapsto x_\a^{-1} y_\b$
for $\a, \b \in \{0,1\}$.
$[(x_\a, y_\b)]$ denotes the equivalent class of $(x_\a, y_\b)$ in $G \times G / G$.
\end{definition}

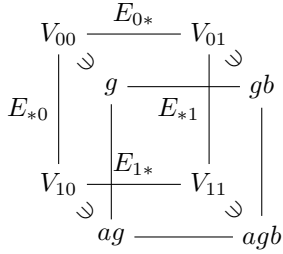
\begin{figure}
  \centering
  \begin{tikzpicture}
    \node (V00) at (0,0) {$V_{00}$};
\node (V10) at (0,-2) {$V_{10}$};
\node (V01) at (2,0) {$V_{01}$};
\node (V11) at (2,-2) {$V_{11}$};

\draw (V00) to node[auto, swap] {$E_{*0}$} (V10);
\draw (V00) to node[auto] {$E_{0*}$} (V01);
\draw (V01) to node[auto, swap] {$E_{*1}$} (V11);
\draw (V10) to node[auto] {$E_{1*}$} (V11);

\draw (0,0)+(0.7,-0.7) node (g) {$g$};
\draw (0,-2)+(0.7,-0.7) node (ag) {$ag$};
\draw (2,0)+(0.7,-0.7) node (gb) {$gb$};
\draw (2,-2)+(0.7,-0.7) node (agb) {$agb$};

\path (g) -- node [sloped] {$\ni$} (V00);
\path (ag) -- node [sloped] {$\ni$} (V10);
\path (gb) -- node [sloped] {$\ni$} (V01);
\path (agb) -- node [sloped] {$\ni$} (V11);

\draw (g) to (ag);
\draw (g) to (gb);
\draw (gb) to (agb);
\draw (ag) to (agb);
  \end{tikzpicture}
  \caption{Left-right bipartite Cayley graph.}
  \label{fig:left-right-cayley-graph}
\end{figure}

The left-right bipartite Cayley graph is illustrated in 
  Figure~\ref{fig:left-right-cayley-graph}.

Note that we labeled the edges to point out the four vertices, $g \in V_{00}, ag \in V_{10}, gb \in V_{01}, agb \in V_{11}$, form a square. The square appears because the left action commutes with the right action. The existence of squares is essential for the Tanner code construction in \cite{panteleev2021asymptotically} and \cite{dinur2021locally}.

\subsection{Chain complex} \label{sec:chain-complex}

Here, we introduce the language of chain complexes from homological algebra. 

\begin{definition}[Chain complex]
A \emph{chain complex} $\mathcal{C}$ is a sequence of vector spaces, $C_i$, together with linear maps, $\partial_i: C_i \rightarrow C_{i-1}$ called the \emph{boundary operators}, where these boundary operators satisfy
\begin{equation}
      \partial_{i-1} \partial_i  = 0.
    \end{equation}
\end{definition}

The kernel and the image of a boundary operator is defined as
$\Ker \partial_i \coloneqq \{c_i \in C_i : \partial_i c_i = 0\}$,
$\Ima \partial_i \coloneqq \{\partial_i c_i \in C_{i-1} : c_i \in C_i\}$.

\section{Main Result: New family of LTCs} \label{sec:distance-LTC}
In this section, we will construct and prove the existence of $c^3$-LTC which are LTCs with constant rate, constant relative distance and constant locality. The construction is based on the balanced product of lossless expander graphs. 

\begin{theorem}[LTC] \label{thm:LTC}
For all $0 < r < 1$, there exist $\delta, s > 0$, $w \in \NN$ and an explicit construction of an infinite family of error-correcting codes $\{C_i\}$ with parameters $[n_i, k_i, d_i]$, 
such that $n_i$ approaches infinity as $i$ increases, $k_i/n_i \ge r, d_i/n_i \ge \delta$ and $C_i$ is $(w,s)$-locally testable. 
\end{theorem}


We first construct the locally testable codes in Sec.~\ref{sec_CON-LTC}, and then prove that the codes have constant rate, linear distance and local testability in Sec.~\ref{sec_ProofLTC}.

\subsection{Construction of LTCs}\label{sec_CON-LTC}

First, we introduce the notations used throughout the section.
The balanced product graph is constructed using two bipartite graphs
$X_{\updownarrow} = (V_{0*}, V_{1*}, E_{\updownarrow})$ and $X_{\leftrightarrow} = (V_{*0}, V_{*1}, E_{\leftrightarrow})$.
We denote the vertices and edges of the balanced product graph by
$V_{00}, V_{10}, V_{01}, V_{11}, E_{*0}, E_{*1}, E_{0*}, E_{1*}$
as shown in Fig.~\ref{fig:hypergraph_BG}. 

The bipartite graphs $X_{\updownarrow}, X_{\leftrightarrow}$ are chosen to satisfy the conditions given in the following theorem.
This can be constructed using the known generalized zig-zag construction for 1-sided lossless expanders \cite{capalbo2002randomness}.

\begin{theorem}[]
  \label{thm:1-sided-lossless-expander-with-symmetry}
  For any $\epsilon > 0$ and interval $I_\updownarrow, 
    I_\leftrightarrow \subseteq (0,1)$,
  there exist parameters
    $(w_\downarrow, w_\uparrow), (w_\rightarrow, w_\leftarrow), 
    (c_\downarrow, \epsilon_\downarrow), (c_\rightarrow, \epsilon_\rightarrow)$,
  such that
  \begin{itemize}
      \item $w_\downarrow/w_\uparrow \in I_\updownarrow$,
      \item $w_\rightarrow/w_\leftarrow \in I_\leftrightarrow$,
      \item $w_\uparrow \epsilon_\rightarrow \le \epsilon$,
      \item $\epsilon_\rightarrow \le \epsilon$ (implied by the last line), 
      \item $\epsilon_\downarrow \le \epsilon$,
  \end{itemize}
  then for any $N \in \NN$,
  there exist group $G$ with $|G|>N$,
  and bipartite graphs $X_\updownarrow = (V_{0*}, V_{1*}, E_\updownarrow), X_\leftrightarrow = (V_{*0}, V_{*1}, E_\leftrightarrow)$
  such that 
  \begin{itemize}
      \item $X_\updownarrow$ is $(w_\downarrow, w_\uparrow)$-regular,
      \item $X_\leftrightarrow$ is $(w_\rightarrow, w_\leftarrow)$-regular,
      \item $X_\updownarrow$ is 1-sided $(c_\downarrow, \epsilon_\downarrow)$-lossless expander,
      \item $X_\leftrightarrow$ is 1-sided $(c_\rightarrow, \epsilon_\rightarrow)$-lossless expander,
      \item $X_\updownarrow, X_\leftrightarrow$ have free $G$-actions,
      \item $|V_{0*}| = \Theta(|G|)$ and $|V_{*0}| = \Theta(|G|)$.
  \end{itemize}
  
\end{theorem}

Notice the conditions on the parameters are not symmetric between $X_{\updownarrow}$ and $X_{\leftrightarrow}$. We discuss this in \ref{rem:skew}.


Now, we are ready to construct the code.
\begin{itemize}
    \item By Theorem~\ref{thm:1-sided-lossless-expander-with-symmetry}, there exists 
    $X_\updownarrow, X_\leftrightarrow$,
    such that the parameters satisfy
    \begin{itemize}
        \item $1-\frac{w_\downarrow}{w_\uparrow}-\frac{w_\rightarrow}{w_\leftarrow} \ge r$,
        \item $|G|=\Theta(|V_{0*}|)=\Theta(|V_{*0}|)$,
        \item $w_\uparrow \epsilon_\rightarrow, \epsilon_\rightarrow, \epsilon_\downarrow \le \epsilon < 1/16$,
    \end{itemize}
    
    \item 
    Next, we take the balanced product $X_\updownarrow \times_G X_\leftrightarrow$
    which gives a 3-term chain complex 
    $\FF_2^{V_{00}}
    \xrightarrow{\partial_2} \FF_2^{V_{10}} \oplus \FF_2^{V_{01}}
    \xrightarrow{\partial_1} \FF_2^{V_{11}}$,
    where $\partial_2(v_{00}) = (L(E_{*0})(v_{00}), L(E_{0*})(v_{00}))$ and $\partial_1((v_{10}, v_{01})) = L(E_{1*})(v_{10}) + L(E_{*1})(v_{01})$ are induced from the adjacency matrices of the balanced product graph.
    Detailed construction is provided in Appendix~\ref{app_Chain_Complex_BP}.
    \item Finally, we obtain a classical code,
    $C(H=\partial_2)$,
    by taking $\partial_2$ as the parity-check matrix,
    where $V_{00}$ are the bits and $V_{10} \cup V_{01}$ are the checks.
    Another way to say it is that $\partial_2$ is the adjacency matrix between $V_{00}$ and $V_{10} \cup V_{01}$.
\end{itemize}

Note that $V_{11}$ is not used explicitly in the code construction.
However, it will be used in the proof of constant soundness.


\subsection{Proof of Theorem~\ref{thm:LTC}}\label{sec_ProofLTC}

Now, we prove that the code, constructed in the previous subsection, has arbitrarily large code length, constant rate, linear distance, constant locality, and constant soundness.

\begin{proof}[Proof of arbitrarily large code length]
  The code length $n = |V_{00}| = |V_{0*}||V_{*0}|/|G| = \Theta(|G|)$.
  By Theorem~\ref{thm:1-sided-lossless-expander-with-symmetry},
  $|G|$ can be arbitrarily large.
  So the code length is arbitrarily large.
\end{proof}
 
\begin{proof}[Proof of constant rate]
There are $n = |V_{00}|$ bits, $m = |V_{10}|+|V_{01}|$ checks, so $k \ge n-m = |V_{00}|-|V_{10}|-|V_{01}|$. Because $X_\updownarrow, X_\leftrightarrow$ are regular bipartite graphs, we know the ratio $|V_{00}| : |V_{10}| : |V_{01}| : |V_{11}| = w_\uparrow w_\leftarrow : w_\downarrow w_\leftarrow
        : w_\uparrow w_\rightarrow : w_\downarrow w_\rightarrow$.
Therefore, the rate, $k/n \ge 1-\frac{w_\downarrow}{w_\uparrow}-\frac{w_\rightarrow}{w_\leftarrow} \ge r$.
\end{proof}
\begin{proof}[Proof of linear distance]
Let $c_2 \in \Ker \partial_2$, i.e. $\partial_2 c_2 = 0$. 
Then $L(E_{*0}) c_2 = 0$. This is sufficient to show linear distance.

From Corollary~\ref{cor:freely-implies-lossless} in Appendix~\ref{app_LTC_Proof}, $(V_{00}, V_{10}, E_{*0})$ is a 1-sided $(c_\downarrow/|V_{*0}/G|, \epsilon_\downarrow)$-lossless expander.
So, if $|c_2| < c_\downarrow/|V_{*0}/G| |V_{00}|$, then $|L(E_{*0}) c_2| > w_\downarrow (1-\epsilon_\downarrow) |c_2|$,
which implies $c_2 = 0$ from $L(E_{*0}) c_2 = 0$.
Therefore, the distance $\ge (c_\downarrow/|V_{*0}/G|) |V_{00}| = \Theta(|V_{00}|) = \Theta(n)$.
\end{proof}

\begin{proof}[Proof of locality]
From the code construction, each check in $V_{10}$ is connected to $w_\uparrow$ bits, and each check in $V_{01}$ is connected to $w_\leftarrow$ bits.
So the code has locality parameter $w = \max(w_\uparrow, w_\leftarrow) = \Theta(1)$.
\end{proof}
\begin{proof}[Proof of constant soundness]
By Lemma~\ref{lem:linear-local-testable-distance-implies-soundness} in Appendix~\ref{app_LTC_Proof} and     $n = \Theta(m)$, it is sufficient to show the chain complex $\mathcal{C}$ has linear locally testable distance (will be defined in Appendix~\ref{app_LTC_Proof}).
By Lemma~\ref{lem:local-minimal-implies-local-testability}, it is sufficient to show the chain complex $\mathcal{C}$ has linear locally minimal distance. By Theorem~\ref{lem:small-set-LTC} and Corollary~\ref{cor:linear-local-minimal-distance}, it is sufficient to check $\epsilon<1/16$ which holds by assumption.
\end{proof}


\section{Conclusion} \label{sec:conclusion}
\subsection{Summary}

In this work, we construct $c^3$-locally testable code using the balanced product \cite{breuckmann2021balanced} of two 1-sided lossless expander graphs \cite{capalbo2002randomness}. This solves the conjecture of the existence of $c^3$-locally testable code. 

\subsection{Discussion}

Here, we compare our construction with two recent papers on similar topics \cite{panteleev2021asymptotically} \cite{dinur2021locally}. 
The common feature of these constructions is that all of them use the same kind of two-dimensional graph. 
The key difference is how one obtains the code from the graph.
Here, we illustrate 3 different ways to obtain a chain complex from the two-dimensional graph.

Take the notation in Definition~\ref{def:left-right-cayley-graph} and Figure~\ref{fig:left-right-cayley-graph}.
Denote vertices $V = V_{00} \cup V_{10} \cup V_{01} \cup V_{11}$,
horizontal edges $E^= = E_{0*} \cup E_{1*}$, 
vertical edges $E^{||} = E_{*0} \cup E_{*1}$,
and faces $F$.

Now, we can compare the chain complex obtained through 3 different methods.
In \cite{panteleev2021asymptotically},
  the chain complex is
  $\FF_2^{E^=} \rightarrow \FF_2^{F} \oplus \FF_2^{V} \rightarrow \FF_2^{E^{||}}$.
In \cite{dinur2021locally},
  the chain complex is
  $\FF_2^{F} \rightarrow \FF_2^{E^{||}} \oplus \FF_2^{E^=} \rightarrow \FF_2^{V}$.
In this work,
  the chain complex is
  $\FF_2^{V_{00}} \rightarrow \FF_2^{V_{10}} \oplus \FF_2^{V_{01}} \rightarrow \FF_2^{V_{11}}$.
For simplicity, we didn't include the base code in the examples above.

\subsection{Future work}

One direction is to consider problems in coding theory similar to LTCs including quantum low-density parity-check codes.
The techniques developed in this work could be applied to these settings.

Another direction is to study the properties of the two-dimensional graphs. 
Given the usefulness of lossless expanders in networks \cite{arora1996line} and computational complexity \cite{ben1999short, alekhnovich2004pseudorandom, alekhnovich2001lower, buresh2003rank}
it may be interesting to look for settings where one needs both structural properties from the chain complex and the expander properties.

\bibliographystyle{IEEEtran}
\bibliography{references.bib}

\appendices

\section{Construction of LTCs}
\subsection{1-sided lossless expander}
\label{app_1-sided-lossless}

In this section, we show Theorem~\ref{thm:1-sided-lossless-expander-with-symmetry}.

It is known that 1-sided lossless expanders exist for arbitrarily small $\epsilon$ through random constructions \cite{hoory2006expander}. However, because we need the graph to have symmetry, we use the explicit construction by \cite{capalbo2002randomness}, where they apply a generalized zig-zag product on Cayley eigenvalue expander graphs. Here is their result. Note that we include additional statements that wasn't stated explicitly in their work.

\begin{theorem}[(modified) 1-sided lossless expander] \label{thm:1-sided-lossless-basic}
  For every $N, T \le N, \epsilon_0 > 0$, 
    there exists a $(D, DT)$-regular bipartite graph $(V_0, V_1, E)$
    with free $G$ action
    which is a 1-sided $(K_{\max}/|V_0|, \epsilon_0)$-lossless expander
    such that,
  \begin{itemize}
    \item $|V_0| > N$,
    \item $D = (\frac{\log T + 1}{\epsilon_0})^{O(1)}$,
    \item $K_{\max} = O(\epsilon_0 |V_0|/T)$,
    \item $|V_0|/|G| = O_{T, \epsilon_0}(1)$,
  \end{itemize}
  where $O_{T, \epsilon_0}(1)$ means it is a constant that depends only on $T$ and $\epsilon_0$.
\end{theorem}



Now, we use this result to show Theorem~\ref{thm:1-sided-lossless-expander-with-symmetry}.

\begin{proof} [Proof of Theorem~\ref{thm:1-sided-lossless-expander-with-symmetry}]
  We find $X_{\updownarrow}$ and $X_{\leftrightarrow}$ in order.
  Namely, we first find $X_{\updownarrow}$ that satisfies
  $w_\downarrow/w_\uparrow \in (\mu_\updownarrow, \mu_\updownarrow')$ and $\epsilon_\downarrow \le \epsilon$.
  Then, we find $X_{\leftrightarrow}$ that satisfies
  $w_\rightarrow/w_\leftarrow \in (\mu_\leftrightarrow, \mu_\leftrightarrow')$ and $\epsilon_\rightarrow \le \epsilon/w_\uparrow$.
  Because $w_\uparrow \ge 1$, $\epsilon_\rightarrow \le \epsilon/w_\uparrow$ implies $\epsilon_\rightarrow \le \epsilon$.

  This can be achieved by the theorem. 
  For $X_{\updownarrow}$, we set $T = w_\uparrow/w_\downarrow \in (\mu_\updownarrow, \mu_\updownarrow'), \epsilon_0 = \epsilon$. 
  For $X_{\leftrightarrow}$, we set $T = w_\leftarrow/w_\rightarrow \in (\mu_\leftrightarrow, \mu_\leftrightarrow'), \epsilon_0 = \epsilon/w_\uparrow$.
  One can see that the desired properties are implied from Theorem~\ref{thm:1-sided-lossless-basic}.
\end{proof}


\subsection{Chain complex from balanced product}
\label{app_Chain_Complex_BP}



In this section, we first provide more detail on the construction of the chain complex, then we prove the chain complex is well defined.

To get the code, we need to obtain vector spaces and linear maps from the graph. We do so by taking the adjacency matrix. 
Recall the adjacency matrix of a bipartite graph
  is a linear map
  $L((V_0, V_1, E)): \FF_2^{V_0} \rightarrow \FF_2^{V_1}$.

Before going further, we introduce a convenient notation.
Because of the one-one correspondence between
  the subsets of $V$ and vectors in $\FF_2^{V}$
  by the map $\v \mapsto \sum_{x \in \v} e_x$
    where $\v \subseteq V$,
  we abuse the notation by overloading both use cases.
For example, we write both $x \in \v$ and $\v \in \FF_2^{V}$,
  where the first $\v$ is interpreted as a set with $x$ as its element
  and the second $\v$ is interpreted as a vector $\v \in \FF_2^{V}$.
Similarly, $x_0$ can be interpreted both as an element in $V_0$
  and the basis vector $e_{x_0}$.
  
Now, we generalize the construction of linear maps from bipartite graphs to
  balanced product of bipartite graphs.
Given a balanced product graph $(V^*, E^*, F)$,
  we obtain the vector spaces 
  $\FF_2^{V_{00}}, \FF_2^{V_{10}}, \FF_2^{V_{01}}, \FF_2^{V_{11}}$
  and the linear maps
  \begin{itemize}
    \item $L_{00 \rightarrow 10}=L(E_{*0}): \FF_2^{V_{00}} \rightarrow \FF_2^{V_{10}}$,
    \item $L_{00 \rightarrow 01}=L(E_{0*}): \FF_2^{V_{00}} \rightarrow \FF_2^{V_{01}}$,
    \item $L_{10 \rightarrow 11}=L(E_{1*}): \FF_2^{V_{10}} \rightarrow \FF_2^{V_{11}}$,
    \item $L_{01 \rightarrow 11}=L(E_{*1}): \FF_2^{V_{01}} \rightarrow \FF_2^{V_{11}}$,
  \end{itemize}
  where
  $L_{00 \rightarrow 10}(x_{00}) = \sum_{x_{10} \in N_{10}(x_{00})} x_{10}$,
  $L_{00 \rightarrow 01}(x_{00}) = \sum_{x_{01} \in N_{01}(x_{00})} x_{01}$,
  $L_{10 \rightarrow 11}(x_{10}) = \sum_{x_{11} \in N_{11}(x_{10})} x_{11}$,
  $L_{01 \rightarrow 11}(x_{01}) = \sum_{x_{11} \in N_{11}(x_{01})} x_{11}$.

These 4 linear maps form a chain complex.
  $\FF_2^{V_{00}}
    \xrightarrow{\partial_2} \FF_2^{V_{10}} \oplus \FF_2^{V_{01}}
    \xrightarrow{\partial_1} \FF_2^{V_{11}}$,
  where
  \begin{equation}
    \partial_2(v_{00}) = (L_{00 \rightarrow 10}(v_{00}), L_{00 \rightarrow 01}(v_{00})),
  \end{equation}
  \begin{equation}
    \partial_1((v_{10}, v_{01})) = L_{10 \rightarrow 11}(v_{10}) + L_{01 \rightarrow 11}(v_{01}).
  \end{equation}
For the chain complex to be well defined,
  we need to check the condition $\partial_1 \partial_2 = 0$,
  which is same as showing 
  $L_{10 \rightarrow 11} L_{00 \rightarrow 10}(v) + L_{00 \rightarrow 01} L_{01 \rightarrow 11}(v) = 0$

The proof consists of 2 steps. We first use the free $G$ action to label the vertices and the edges of $X, Y$.
Then using the labeling, we express the linear maps explicitly and show the equality.

First, we show a set $V$, with free $G$ action satisfies $V = G \times R$,
  with the $G$ action acting on the component $G$.
For each orbit $G x$, we pick a representative element $r$.
Let $R$ be the set of the representative elements for all orbits.
Then each element $x \in V$ can be written uniquely as $g r$ for some $g \in G, r \in R$.
This induces the bijection between $V$ and $G \times R$.
Furthermore $g' (g r) = (g'g) r$, which implies
  the $G$-action on $V$ is the same $G$-action on the component $G$ of $G \times R$. 

Next, we show the edges can be decomposed into Cayley graphs.
Given a bipartite graph $(V_0, V_1, E)$ with left free $G$-action, 
  we can label the vertices by
  $V_0 = \{(g_0, r_0) : g_0 \in G, r_0 \in R_0\}$
  and $V_1 = \{(g_1, r_1) : g_1 \in G, r_1 \in R_1\}$.
Now, consider the subgraph $X_{r_0, r_1}$ between the vertices
  $\{(g_0, r_0) : g_0 \in G\}$ and $\{(g_1, r_1) : g_1 \in G\}$
  for each $r_0 \in R_0$ and $r_1 \in R_1$.
It is clear that $X = \cup_{r_0 \in R_0, r_1 \in R_1} X_{r_0, r_1}$.
Now, we show each $X_{r_0, r_1}$ is a right bipartite Cayley graph.
  Let $A_{r_0, r_1} = \{g_0^{-1}g_1 : ((g_0, r_0), (g_1, r_1)) \in E\}$.
  Because the graph is left-$G$ invariant, 
    $((g_0, r_0), (g_1, r_1)) \in X_{r_0, r_1} \Rightarrow
    ((gg_0, r_0), (gg_1, r_1)) \in X_{r_0, r_1}$.
  Therefore, $X_{r_0, r_1} = \Gamma_{\textnormal{right}}(G, A_{r_0, r_1})$.

Finally, we will use the labels above to express the vector spaces and the linear maps.
Consider the balanced product graph $X \times_G Y$,
  where $X = (G \times R_0, G \times R_1, 
    \cup_{r_0 \in R_0, r_1 \in R_1} X_{r_0, r_1})$,
    $X_{r_0, r_1} = \Gamma_{\textnormal{right}}(G, A_{r_0, r_1})$
    and $Y = (G \times S_0, G \times S_1, 
    \cup_{s_0 \in S_0, s_1 \in S_1} Y_{s_0, s_1})$,
    $Y_{s_0, s_1} = \Gamma_{\textnormal{right}}(G, B_{s_0, s_1})$.

For the vector spaces, for each $\a, \b \in \{0,1\}$ we have
\begin{equation}
  V_{\a\b} = \{(g, r_\a, g', s_\b): g, g' \in G, r_\a \in R_\a, s_\b \in S_\b\} / \sim
\end{equation}
where $(g, r_\a, g', s_\b) \sim (g''g, r_\a, g''g', s_\b)$.
This implies an alternative labeling
  $V_{\a\b} = \{(h, r_\a, s_\b): h \in G, r_\a \in R_\a, s_\b \in S_\b\}$,
  where $(g, r_\a, g', s_\b) \mapsto (g^{-1}g', r_\a, s_\b)$.

For the linear maps, by definition, we have
\begin{equation}
  L_{00 \rightarrow 10}((h, r_0, s_0))
  = \sum_{r_1 \in R_1}\sum_{a \in A_{r_0, r_1}} (a^{-1}h, r_1, s_0),
\end{equation}
\begin{equation}
  L_{00 \rightarrow 01}((h, r_0, s_0))
  = \sum_{s_1 \in S_1}\sum_{b \in B_{s_0, s_1}} (hb, r_0, s_1),
\end{equation}
\begin{equation}
  L_{01 \rightarrow 11}((h, r_0, s_1))
  = \sum_{r_1 \in R_1}\sum_{a \in A_{r_0, r_1}} (a^{-1}h, r_1, s_1),
\end{equation}
\begin{equation}
  L_{10 \rightarrow 11}((h, r_1, s_0))
  = \sum_{s_1 \in S_1}\sum_{b \in B_{s_0, s_1}} (hb, r_1, s_1).
\end{equation}

So,
\begin{align}
  & \partial_1 \partial_2 ((h, r_0, s_0))\\
  &= L_{10 \rightarrow 11} L_{00 \rightarrow 10}((h, r_0, s_0)) + L_{01 \rightarrow 11} L_{00 \rightarrow 01}((h, r_0, s_0)) \\
  &= \sum_{s_1 \in S_1}\sum_{b \in B_{s_0, s_1}}\sum_{r_1 \in R_1}\sum_{a \in A_{r_0, r_1}} (a^{-1}hb, r_1, s_1) \\
  &+ \sum_{r_1 \in R_1}\sum_{a \in A_{r_0, r_1}}\sum_{s_1 \in S_1}\sum_{b \in B_{s_0, s_1}} (a^{-1}hb, r_1, s_1) \\
  &= 0
\end{align}




\section{Technical lemmas for  Theorem~\ref{thm:LTC}}
\label{app_LTC_Proof}
In this section, we provide relevant materials for the proof of main theorem, Theorem~\ref{thm:LTC}.

We first prove a simple lemma that shows the one-dimensional subgraph of the balanced product graph remains a lossless expander.
Then we introduce the locally testable distance and the locally minimal distance.

\begin{lemma} \label{lem:freely-implies-copy}
Given two graphs $X_{\updownarrow} = (V_{0*}, V_{1*}, E_{\updownarrow})$ and $X_{\leftrightarrow} = (V_{*0}, V_{*1}, E_{\leftrightarrow})$ with free $G$-invariant action. Then for the one-dimensional subgraph of $X_{\updownarrow} \times_G X_{\leftrightarrow}$,
$(V_{00}, V_{10}, E_{*0})$ is isomorphic to $|V_{*0}/G|$ copies of $X_{\updownarrow}$. Similarly,  
$(V_{01}, V_{11}, E_{*1})$ is isomorphic to $|V_{*1}/G|$ copies of $X_{\updownarrow}$,
$(V_{00}, V_{01}, E_{0*})$ is isomorphic to $|V_{0*}/G|$ copies of $X_{\leftrightarrow}$, and 
$(V_{10}, V_{11}, E_{1*})$ is isomorphic to $|V_{1*}/G|$ copies of $X_{\leftrightarrow}$.
\end{lemma}

\begin{proof}
  Here, we show the case for $(V_{00}, V_{10}, E_{*0})$.
    Other cases follow similarly.
  Before quotienting,
    the 1d subgraph $(V'_{00}, V'_{10}, E'_{*0})$ of
    the hypergraph product $X_{\updownarrow} \times X_{\leftrightarrow}$
    is isomorphic to $|V_{*0}|$ copies of $X_{\updownarrow}$
    each labeled by $V_{*0}$.
  After quotienting,
    the copies in the same orbit of $V_{*0}$
    are identified into 1 copy of $X_{\updownarrow}$.
    So we are left with $|V_{*0}/G|$ copies of $X_{\updownarrow}$.
\end{proof}

\begin{corollary} \label{cor:freely-implies-lossless}
  Under the same assumption in \ref{lem:freely-implies-copy}.
  Futhermore, $X_{\updownarrow}$ is a 1-sided $(c_\downarrow, \epsilon_\downarrow)$-lossless expander, $X_{\leftrightarrow}$ is a 1-sided $(c_\rightarrow, \epsilon_\rightarrow)$-lossless expander.
  Then $(V_{00}, V_{10}, E_{*0})$ is a 1-sided $(c_\downarrow V_{0*}/V_{00}, \epsilon_\downarrow)$-lossless expander, 
  $(V_{01}, V_{11}, E_{*1})$ is a 1-sided $(c_\downarrow V_{0*}/V_{01}, \epsilon_\downarrow)$-lossless expander,
  $(V_{00}, V_{01}, E_{0*})$ is a 1-sided $(c_\rightarrow V_{*0}/V_{00}, \epsilon_\rightarrow)$-lossless expander,
  $(V_{10}, V_{11}, E_{1*})$ is a 1-sided $(c_\rightarrow V_{*0}/V_{10}, \epsilon_\rightarrow)$-lossless expander.
\end{corollary}

\begin{proof}
  Here, we show the case for $(V_{00}, V_{10}, E_{*0})$.
    Other cases follow similarly.
  By definition, any small set $v_{0*} \subseteq V_{0*}$, $|v_{0*}| < c_\downarrow |V_{0*}|$,
    satisfies $|N_{V_{1*}}(v_{0*})| \ge (1-\epsilon_\downarrow) w_\downarrow |v_{0*}|$.

  Now, given a small set $v_{00} \subseteq V_{00}$, $|v_{00}| < c_\downarrow |V_{0*}|$.
    Let $v_{00} = \cup_{i=1}^{|V_{*0}/G|} v_{00, i}$, where $v_{00, i}$ is the intersection of $v_{00}$ with the $i$-th copy of $X_{\updownarrow}$.
  Because each $v_{00, i}$ is small, $|v_{00, i}| < c_\downarrow |V_{0*}|$,
    the size of its neighbor $|N_{V_{10}}(v_{00, i})|$ has size at least $(1-\epsilon_\downarrow)|v_{00, i}|$
  Therefore, $|N_{V_{10}}(v_{00})| = \sum_{i=1}^{|V_{*0}/G|} |N_{V_{10}}(v_{00, i})|
    \ge \sum_{i=1}^{|V_{*0}/G|} (1-\epsilon_\downarrow)|v_{00, i}| = (1-\epsilon_\downarrow)|v_{00}|$.
  This implies $(V_{00}, V_{10}, E_{*0})$ is a 1-sided $(c_\downarrow V_{0*}/V_{00}, \epsilon_\downarrow)$-lossless expander.
\end{proof}

\subsection{Locally Testable Distance}



Locallly testable distance is a parameter that is related to soundness of the LTC. If a code has linear locally testable distance, then the code has constant soundness.


\begin{definition}[Locally testable distance]
Given is a chain complex $\mathcal{C}: C_{i+1} \xrightarrow{\partial_{i+1}} C_{i}$. The locally testable distance $d_i^{LT}(\mathcal{C})$ is the maximal value such that for any short vector $c_i \in \Ima \partial_{i+1}$ with $|c_i| < d_i^{LT}(\mathcal{C})$, there exists a short vector $c_{i+1} \in C_{i+1}$ such that $c_i = \partial_{i+1} c_{i+1}$ and $|c_{i+1}| = O(|c_i|)$.
\end{definition}


\begin{lemma}[Linear locally testable distance implies constant soundness]
  \label{lem:linear-local-testable-distance-implies-soundness}
  Given a chain complex 
    $C_2 \xrightarrow{\partial_2} C_1$.
  The code $C=C(H)$ defined by taking $H = \partial_2$ as the parity check matrix,
    where $\FF_2^n = C_2$ are the bits and $\FF_2^m = C_1$ are the checks,
    satisfies 
    \begin{equation}
      \frac{1}{m} |Hx| \ge \frac{s}{n} d(x, C(H)),
    \end{equation}
    where $s =\min(O(\frac{n}{m}), \frac{d_1^\text{LT}}{m})$.
  
  Therefore, if $d_1^\text{LT} = \Theta(m) = \Theta(n)$,
    then $s = \Theta(1)$, i.e. constant soundness.
\end{lemma}

\begin{proof}
  Equivalently, we show for all $c_1 \in \Ima \partial_2$,
    there exists $c_2 \in C_2$,
    such that $c_1 = \partial_2 c_2$ and 
    $\frac{1}{m} |c_1| \ge \frac{s}{n} |c_2|$.
  We split into 2 cases.
  If $|c_1| < d_1^\text{LT}$, 
    by the definition of locally testable distance,
    there exists $c_2$ with $|c_2| = O(|c_1|)$.
    In this case, $\frac{1}{m} |c_1| \ge \frac{O(n/m)}{n} |c_2|$.
  Otherwise, $|c_1| \ge d_1^\text{LT}$. 
    Because $c_1 \in \Ima \partial_2$,
    there exists $c_2$ with $|c_2| \le n$.
    In this case, $\frac{1}{m} |c_1| \ge \frac{d_1^\text{LT}/m}{n} |c_2|$.
  Overall, $\frac{1}{m} |c_1| \ge \frac{s}{n} |c_2|$
  with $s = \min(O(\frac{n}{m}), \frac{d_1^\text{LT}}{m})$.
\end{proof}

\subsection{Locally minimal} \label{sec:local-minimal}

In this section, we introduce the locally minimal distance
  \cite{evra2020decodable}
  \cite{kaufman2014ramanujan}
  \cite{panteleev2021asymptotically},
  and show that locally minimal distance is a lower bound of
  locally testable distance.
This means that if we can show locally minimal distance is linear
  then locally testable distance is linear which implies constant soundness.

\begin{definition}[Locally minimal]
  Given a chain complex 
    $C_{i+1} \xrightarrow{\partial_{i+1}} C_{i}$.
  A vector $c_i \in C_i$ is locally minimal if 
    for any basis vector $e_{i+1} \in C_{i+1}$
    \begin{equation}
      |c_i + \partial_{i+1} e_{i+1}| \ge |c_i|.
    \end{equation}
\end{definition}

The definition of locally minimal is related to the greedy flipping decoder of the expander code.

\begin{definition}[Greedy flipping algorithm]
  Input: $c_i \in C_i$.
  \begin{enumerate}
    \item If there exists a basis vector $e_{i+1} \in C_{i+1}$,
      such that $|c_i + \partial_{i+1} e_{i+1}| < |c_i|$,
      replace $c_i$ with $c_i + \partial_{i+1} e_{i+1}$. 
    \item Repeat until no such $e_{i+1}$ exists. Output $c_i$.
  \end{enumerate}
\end{definition}

Any output of a greedy flipping algorithm is locally minimal.
Note that $c_i$ strictly decreases in each iteration.
So the algorithm halts in $|c_i|$ steps.
Note that $\partial_{i+1} c_i$ does not change throughout the algorithm
  because the change of $\partial c_i$ is $\partial \partial e_i = 0$.
We refer the process of 
  replacing $c_i$ with $c_i + \partial_{i+1} e_{i+1}$ flipping, 
  because in $\FF_2$, the bits flip between 0 and 1.

In our context, we consider a modification, namely, the weighted locally minimal,
  where we normalize the weight
    before comparing $c_i + \partial_{i+1} e_{i+1}$ and $c_i$.
  The purpose of performing this additional normalization is to make the statement more natural.

The normalization is determined through the following discussion.
For a chain complex constructed from balanced product of regular bipartite graphs,
  $\partial e_2$ flips 
  $w_\downarrow$ bits in $\FF_2^{V_{10}}$
  and $w_\rightarrow$ bits in $\FF_2^{V_{01}}$.
So we will weight the components in $\FF_2^{V_{10}}$ with $1/w_\downarrow$,
  and the components in $\FF_2^{V_{01}}$ with $1/w_\rightarrow$.

\begin{definition}[Weighted locally minimal]
  Given a chain complex 
    $C_2 \xrightarrow{\partial_2} C_1 
    \xrightarrow{\partial_1} C_0$
    constructed from balanced product of regular bipartite graphs.
  A vector $c_1 = (v_{10}, v_{01}) \in \Ker \partial_1$ is weighted locally minimal if 
    for any basis vector $e_2 \in C_2$
    \begin{equation}
      |c_1 + \partial e_2|_w \ge |c_1|_w,
    \end{equation}
    where $|(v_{10}, v_{01}))|_w = |v_{10}|/w_\downarrow + |v_{01}|/w_\rightarrow$.
\end{definition}

The greedy flipping algorithm still applies.
Because in each step, $|c_1|_w$ strictly decreases by at least $1/\max(w_\downarrow, w_\rightarrow)$,
  the algorithm halts in $|c_1|_w \max(w_\downarrow, w_\rightarrow) \le |c_1| \max(w_\downarrow, w_\rightarrow)/\min(w_\downarrow, w_\rightarrow) = \Theta(|c_1|)$ steps.

From now on, we only consider the chain complex
  constructed from balanced product of regular bipartite graphs,
  and locally minimal always means weighted locally minimal.

Now, we define the locally minimal distance.

\begin{definition}[Locally minimal distance]
  Given a chain complex 
    $\mathcal{C}: C_2 \xrightarrow{\partial_2} C_1 
    \xrightarrow{\partial_1} C_0$.
  The locally minimal distance $d_1^{LM}(\mathcal{C})$ is 
    the smallest weight of the non trivial locally minimal vectors.
  Formally,
  \begin{equation}
    d_1^{LM}(\mathcal{C}) = \min_{c_1 \in \Ker (\partial_1), c_1\text{ is locally minimal}, c_1 \ne 0} |c_1|.
  \end{equation}
\end{definition}

Finally, we show locally minimal distance is a lower bound of
  locally testable distance.

\begin{lemma}[Linear locally minimal distance implies linear locally testable distance]
  \label{lem:local-minimal-implies-local-testability}
  Given a chain complex 
    $\mathcal{C}: C_2 \xrightarrow{\partial_2} C_1 
    \xrightarrow{\partial_1} C_0$.
  Then 
  \begin{equation}
    d_1^{LT}(\mathcal{C}) \ge d_1^{LM}(\mathcal{C}).
  \end{equation}
\end{lemma}

\begin{proof}
  Recall locally minimal distance means
    for any $c_1 \in \Ker \partial_1$ with $|c_1| < d_1^{LM}(\mathcal{C})$,
    greedy flipping algorithm returns the 0 vector
    and halt in $\Theta(|c_1|)$ steps.
  This means by summing up the basis vector used in each step,
    we obtain a vector $c_2 \in C_2$,
    such that $\partial c_2 = c_1$,
    and $|c_2| \le \Theta(|c_1|)$.
  This satisfies the criteria of strong locally testable,
    therefore, $d_1^{LT}(\mathcal{C}) \ge d_1^{LM}(\mathcal{C})$.
\end{proof}





\section{Proof of small set LTC lemma}

The proof of LTC relies on a key lemma which we call the small set LTC lemma.
The small set LTC lemma implies the chain complex
  has properties similar to LTC,
  namely the number violation in the constraint is larger than the weight of the input vector,
  with the additional assumption that the weight of the input vector is small.

\begin{lemma}[Small set LTC] \label{lem:small-set-LTC}
  Consider the chain complex, $C_2 \xrightarrow{\partial_2} C_1 \xrightarrow{\partial_1} C_0$, 
    constructed from the balanced product graph
    $X^{\updownarrow} \times_G X^{\leftrightarrow} = (V^*, E^*, F)$,
  where $X^{\updownarrow}$ and $X^{\leftrightarrow}$ satisfy the conditions in Theorem
    \ref{thm:1-sided-lossless-expander-with-symmetry}.

If $c_0 = \partial{c_1}$ for some short locally minimal $c_1 = (v_{10}, v_{01}) \in C_1$, 
    with $|v_{10}| < \min (c_\downarrow |V_{0*}|/w_\uparrow, c_\rightarrow |V_{*0}|), \\
    |v_{01}| < \min (c_\rightarrow |V_{*0}|/w_\leftarrow, c_\downarrow |V_{0*}|)$,
  then
  \begin{equation}
    (\frac{1}{2} - 8 \epsilon)|c_1|_w \le |c_0|_w,
  \end{equation}
    where $|c_1|_w = |v_{10}|/w_\downarrow + |v_{01}|/w_\rightarrow$
    and $|c_0|_w = |c_0|/(w_\downarrow w_\rightarrow)$. 
\end{lemma}


Small set LTC is used to show linear locally minimal distance and thus showing LTC.


\begin{corollary}[Small set LTC implies linear locally minimal distance] \label{cor:linear-local-minimal-distance}
  Under the same assumption as in the lemma~\ref{lem:small-set-LTC}
    and $\epsilon < 1/16$,
    we have
  \begin{equation}
    d_1^{LM}(\mathcal{C}) \ge \min (c_\downarrow|V_{0*}|/w_\uparrow, c_\rightarrow|V_{*0}|, 
      c_\rightarrow|V_{*0}|/w_\leftarrow, c_\downarrow|V_{0*}|).
  \end{equation}
\end{corollary}

Because $\Theta(|V_{0*}|) = \Theta(|V_{*0}|) = |V_{00}| = n$,
  and $w_\downarrow$, $w_\uparrow$, $w_\rightarrow$, $w_\leftarrow$, $c$ are $\Theta(1)$,
  we have $d_1^{LM}(\mathcal{C}) \ge \Theta(n)$
which means the locally minimal distance is linear.

\begin{proof} [Proof of Corollary~\ref{cor:linear-local-minimal-distance}]
  Recall the definition of locally minimal distance,
    $d_1^{LM}(\mathcal{C}) = \min_{c_1 \in Ker \partial_1, c_1\text{ is locally minimal}, c_1 \ne 0} |c_1|$.

  By lemma~\ref{lem:small-set-LTC} we know if 
    $|v_{10}| < \min(c|V_{0*}|/w_\leftarrow, c|V_{*0}|)$ and 
    $|v_{01}| < \min(c|V_{*0}|/w_\uparrow, c|V_{0*}|)$
    then $(\frac{1}{2} - 8 \epsilon) |c_1|_w \le |c_0|_w = 0$,
    for $c_1 \in Ker \partial_1, c_1\text{ is locally minimal}$.
  Because $\epsilon < 1/16$, this implies $|c_1|_w = 0$, which means $c_1 = 0$.
  Therefore, if $c_1 \ne 0$,
    at least one of 
      $|v_{10}| < \min(c|V_{0*}|/w_\leftarrow, c|V_{*0}|), 
      |v_{01}| < \min(c|V_{*0}|/w_\uparrow, c|V_{0*}|)$
    is violated.
  So $|c_1| \ge \min (c_\downarrow|V_{0*}|/w_\uparrow, c_\rightarrow|V_{*0}|, 
  c_\rightarrow|V_{*0}|/w_\leftarrow, c_\downarrow|V_{0*}|)$.
\end{proof}

Now, we comment on the optimality our result.
Note that the parameter for the ratio between 
  $|c_1|_w$ and $|c_0|_w$ approaches $\frac{1}{2}$
  as $\epsilon$ approaches 0.
We give an example that shows $\frac{1}{2}$ is optimal.

\begin{example}
  Assume $w_\downarrow, w_\rightarrow$ are both even.
  Pick a vertex $x_{00} \in V_{00}$.
  Pick $n_{10} \subseteq N_{V_{10}}(x_{00})$ and $n_{01} \subseteq N_{V_{01}}(x_{00})$
    of size $w_\downarrow/2$ and $w_\rightarrow/2$.
  
  Set $c_1 = (n_{10}, n_{01})$ which is locally minimal. 
  Because each flip $\partial e_2$, 
    flips $w_\downarrow$ bits in $V_{10}$
    and $w_\rightarrow$ bits in $V_{01}$,
    no flip $\partial e_2$ can reduce the weight $|c_1|_w$.
  Now, $c_0 = \partial c_1 = n_{10} \times_{x_{00}} (N_{V_{01}}(x_{00}) - n_{01}) \cup (N_{V_{10}}(x_{00}) - n_{10}) \times_{x_{00}} n_{01}$,
    so $|c_0| = w_\downarrow w_\rightarrow / 2$.
  The notation $\times_{x_{00}}$ is discussed immediately in \ref{sec:lemma-balanced-product}
  Overall, $|c_1|_w = 1$ and $|c_0|_w = 1/2$.
\end{example}

\begin{figure}
  \centering
  \begin{tikzpicture}
    \def\eps{2pt}

\fill (0,0) circle[radius=1pt] node [label=above left:$x_{00}$] {};

\filldraw[fill=black, draw=black] (-\eps,-2) rectangle (\eps,-4);
\filldraw[fill=white, draw=black] (-\eps,-4) rectangle (\eps,-6);

\filldraw[fill=black, draw=black] (2,-\eps) rectangle (4,\eps);
\filldraw[fill=white, draw=black] (4,-\eps) rectangle (6,\eps);

\filldraw[fill=white, draw=black] (2,-2) rectangle (4,-4);
\filldraw[fill=white, draw=black] (4,-4) rectangle (6,-6);
\filldraw[fill=black, draw=black] (2,-4) rectangle (4,-6);
\filldraw[fill=black, draw=black] (4,-2) rectangle (6,-4);

\path (-\eps,-2) -- node[anchor=south, rotate = 90]{$n_{10}(x_{00})$} (-\eps,-4);
\path (-\eps,-4) -- node[anchor=south east, rotate = 90]{$N_{V_{10}}(x_{00}) - n_{10}(x_{00})$} (-\eps,-4);

\path (2,\eps) -- node[anchor=south]{$n_{01}(x_{00})$} (4,\eps);
\path (4,\eps) -- node[anchor=south west]{$N_{V_{01}}(x_{00}) - n_{01}(x_{00})$} (4,\eps);

  \end{tikzpicture}
  \caption{Sharp example with $|c_0|_w = \frac{1}{2} |c_1|_w$. Black means the value is 1. White means the value is 0.}
  \label{fig:sharp-example-small-set-LTC}
\end{figure}
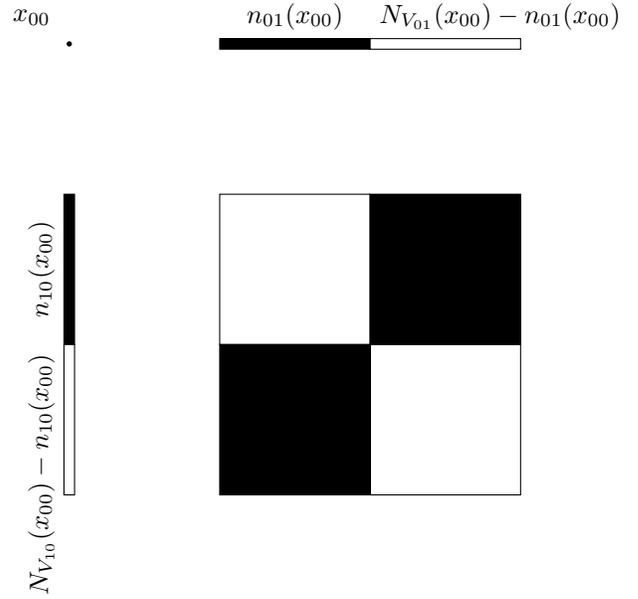

\subsection{Lemma for balanced product} \label{sec:lemma-balanced-product}

Here, we use the labeling developed in \ref{app_Chain_Complex_BP} to show whenever we have 3 vertices that form a wedge,
there exists a vertex that completes the wedge into a square.

\begin{lemma}[Square completion] \label{lem:square-completion}
  Given the balanced product graph $(V^*, E^*, F)$,
  for any $x_{00} \in V_{00}, x_{10} \in V_{10}, x_{01} \in V_{01}$,
  where $(x_{00}, x_{10}) \in E_{*0}, (x_{00}, x_{01}) \in E_{0*}$
  there exists a unique $x_{11} \in V_{11}$,
  such that $(x_{00}, x_{10}, x_{01}, x_{11}) \in F$.
\end{lemma}

We denote the vertex $x_{11} = x_{10} \times_{x_{00}} x_{01}$.
This notation also applies to the sets $v_{10} \times_{x_{00}} v_{01} = \{x_{10} \times_{x_{00}} x_{01} : x_{10} \in v_{10}, x_{01} \in v_{01}\}$. 
Note that this notation only make sense when all vertices in $v_{10}$ and $v_{01}$ are neighbor to $x_{00}$.

Because the construction is symmetric,
  the square completion property also holds for other combination of 3 vertices.
For example, for any $x_{11} \in V_{11}, x_{10} \in V_{10}, x_{01} \in V_{01}$,
where $(x_{10}, x_{11}) \in E_{1*}, (x_{01}, x_{11}) \in E_{*1}$
there exists a unique $x_{00} \in V_{00}$,
such that $(x_{00}, x_{10}, x_{01}, x_{11}) \in F$.

\begin{proof}
  We use the labeling in \ref{app_Chain_Complex_BP}.
  Because $(x_{00}, x_{10}) \in E_{*0}, (x_{00}, x_{01}) \in E_{0*}$,
  we can write $x_{00} = (h, r_0, s_0), x_{10} = (a^{-1}h, r_1, s_0), x_{01} = (hb, r_0, s_1)$, where $h \in G, a \in A_{r_0, r_1}, b \in B_{s_0, s_1}$.
  Now, set $x_{11} = (a^{-1}hb, r_1, s_1)$.
  Then, $(x_{00}, x_{10}, x_{01}, x_{11}) \in F$ forms a square.
\end{proof}

\subsection{Lemma for lossless expanders}

Here, we state two technical lemmas for lossless expanders.
The first is a simple lemma about a lower bound on the number of unique neighbors.
The second lemma is to make quantitative statement 
  about the intuition that lossless expanders look like trees.

Given a bipartite graph $(V_0, V_1, E)$.
If a vertex $x_1 \in V_1$ is neighbor to exactly one vertex in $v_0 \subseteq V_0$,
  we say $x_1$ is a unique neighbor of $v_0$.
We denote the unique neighbor of $v_0$ in $V_1$ as $N^{\textnormal{unique}}_{V_1}(v_0)$.

The first lemma says a lossless expander graph has many unique neighbors.
\begin{lemma} [Lossless expander implies unique expander] \label{lem:unique-expander}
  Let $(V_0, V_1, E)$ be a $(w_0, w_1)$-regular bipartite graph and $(c, \epsilon)$-lossless expander.
  Then for small set $v_0 \subseteq V_0$, $|v_0| < c |V_0|$,
  we have $|N^{\textnormal{unique}}_{V_1}(v_0)| \ge (1-2\epsilon) w_0 |v_0|$.
\end{lemma}

\begin{proof}
  Let $a_i$ be the number of $x_1$ with $\deg_{v_0}(x_1) = i$.
  
  By the definition of lossless expanders, $|N_{V_1}(v_0)| \ge (1-\epsilon) w_0 |v_0|$.
  So $\sum_{i=1}^{w_1} a_i \ge (1-\epsilon) w_0 |v_0|$.

  By counting the edges between $v_0$ and $N_{V_1}(v_0)$ in two ways,
  we have $w_0 v_0 = \sum_{i=1}^{w_1} i a_i$.

  Together, $a_1 \ge 2 \sum_{i=1}^{w_1} a_i - \sum_{i=1}^{w_1} i a_i \ge (1-2\epsilon) w_0 |v_0|$,
  where we use $a_i \ge 0$ in the first inequality.
\end{proof}

Now, we turn to the second lemma.

First, we recall some definitions.
A multiset is a modification of the concept of a set,
  where it is allowed for multiple instances.
For example, $\{a, a, b\}$ is a multiset
  with 2 instances of $a$ and 1 instance of $b$.
Another way to represent a multiset is to
  indicate the number of instances on the upper indices.
For example, $\{a, a, b\}$ can be written as $\{a^2, b\}$.

We say a multiset $A$ majorizes another multiset $B$
  if $\sum_{i=1}^k a_i \ge \sum_{i=1}^k b_i$
  for all $k = 1, ..., i_{\max}$
  where $i_{\max}=|A|=|B|$
    and $\{a_i\}_{i=1}^{i_{\max}}$ and $\{b_i\}_{i=1}^{i_{\max}}$
    are sorted sequences of $A$ and $B$ in the descending order.
This is denoted as $A \succeq B$.
In the case where $A$ and $B$ have different number of elements
  we append 0s so that they become the same size
  and can be compared.

Here is the lemma.
\begin{lemma}[] \label{lem:lossless-split}
  Given a $(w_0, w_1)$ regular bipartite graph, $(V_0, V_1, E)$,
    with 1-sided lossless expansion $(c, \epsilon)$.
  For any small subset $v_1 \subseteq V_1$ with $|v_1| < c|V_0|/w_1$,
    we can split $\deg_{v_1}(x_0)$ into 2 contributions, 
    $\deg_{v_1}(x_0) = d^1(x_0) + d^2(x_0)$,
    such that

  \begin{equation} \label{eq:bound-da}
    \sum_{x_0 \in V_0} d^1(x_0) \le |v_1|,
  \end{equation}

  \begin{equation} \label{eq:bound-db}
    \{d^2(x_0)\}_{x_0 \in V_0}
    \preceq \{\epsilon w_0^{\lceil \frac{w_1}{\epsilon w_0} |v_1| \rceil}\},
  \end{equation}

\end{lemma}


To show the lemma, we use result.
\begin{lemma} \label{lem:lossless-ineq}
  Given a $(w_0, w_1)$-regular bipartite graph, $(V_0, V_1, E)$,
    with 1-sided $(c, \epsilon)$-lossless expansion from $V_0$ to $V_1$.
  For any subsets $v_0 \subseteq V_0$ and $v_1 \subseteq V_1$,
    with $|v_0| < c|V_0|$,
    the number of edges between $v_0$ and $v_1$, $|E(v_0, v_1)|$,
    is bounded by
  \begin{equation} \label{eq:ineq-lossless-v1}
    |E(v_0, v_1)| = \sum_{x_0 \in v_0} \deg_{v_1}(x_0) \le \epsilon w_0 |v_0| + |v_1|,
  \end{equation}
\end{lemma}

\begin{proof}
  We prove the inequality by 
    considering the graph consists of
      $v_0$ and its neighbors $N_{V_1}(v_0)$,
    then remove vertices $N_{V_1}(v_0)-v_1$ and 
    the edges connected to those vertices.

  First, the graph form by $v_0$ and $N_{V_1}(v_0)$ has $w_0 |v_0|$ edges.
  Next, we remove $N_{V_1}(v_0)-v_1$.
  By doing so,
    we remove at least $|N_{V_1}(v_0)|-|v_1|$ edges.
  So 
  \begin{equation*}
    |E(v_0, v_1)| \le w_0 |v_0| - (|N_{V_1}(v_0)|-|v_1|) \le \epsilon w_0 |v_0| + |v_1|,
  \end{equation*}
    where the last inequality follows from the lossless assumption
    $|N_{V_1}(v_0)| \ge (1-\epsilon) w_0 |v_0|$.
\end{proof}

\begin{proof} [Proof of Lemma~\ref{lem:lossless-split}]
  By setting $d^1(x_0) = \max(\deg(x_0)-\epsilon w_0, 0)$, 
    $d^2(x_0) = \deg_{v_1}(x_0) - d^1(x_0)$,
    we claim they satisfy the inequalities.

  First, we show $\sum_{x_0 \in V_0} d^1(x_0) \le |v_1|$.
  Let $v_0 = \{x_0 \in V_0: \deg_{v_1}(x_0) > \epsilon w_0\}$.
  Then,
  \begin{align}
    & \sum_{x_0 \in V_0} d^1(x_0) \\
    &= \sum_{x_0 \in v_0} (\deg_{v_1}(x_0)-\epsilon w_0) \\
    &\le |v_1|,
  \end{align}
  where the last inequality follows from the lemma~\ref{lem:lossless-ineq},
    $\sum_{x_0 \in v_0} \deg_{v_1}(x_0) \le \epsilon w_0 |v_0| + |v_1|$.

  Next, we show $\{d^2(x_0)\}_{x_0 \in V_0}
    \preceq \{\epsilon w_0^{\lceil \frac{w_1}{\epsilon w_0} |v_1| \rceil}\}$.
  Because $d^2(x_0) \le \epsilon w_0$
    and $\sum_{x_0 \in V_0} d^2(x_0) \le \sum_{x_0 \in V_0} \deg_{v_1}(x_0) = w_1 |v_1|$,
    this implies the desired majorization.
\end{proof}

\subsection{Prove the small set LTC lemma}

We first describe the proof idea.

The goal is to show there are a large number of vertices in $V_{11}$
  that are uniquely connected to $v_{01} \cup v_{10}$.
When a vertex $x_{11} \in V_{11}$ is uniquely connected,
  $\partial c_1$ is non zero at $x_{11}$.

To show that most of the vertices are uniquely connected, we observe the following.
  Because $(V_{01}, V_{11}, E_{*1})$ is a lossless expander, by Corollary~\ref{cor:freely-implies-lossless} in Appendix~\ref{app_LTC_Proof},
  most neighbors of $v_{01}$ in $V_{11}$ has unique neighbor in $v_{01}$.
Similarly, most neighbors of $v_{10}$ in $V_{11}$ has unique neighbor in $v_{10}$.
That means, most vertices in $v_{11}$ is neighbor to 
  at most one vertex in $v_{01}$ and one vertex in $v_{10}$.
All these vertices are uniquely connected to $v_{01} \cup v_{10}$,
  unless the vertices are neighbor to both $v_{01}$ and $v_{10}$.
So the last thing we need to estimate is the number of vertices that is
  unique neighbor to both $v_{01}$ and $v_{10}$.

Now, we use the squares in the balanced product graph.
  By Lemma~\ref{lem:square-completion},
    if $x_{11}$ is neighbor to $x_{01} \in v_{01}$ and $x_{10} \in v_{10}$,
  then there is a vertex $x_{00} \in V_{00}$
  where $(x_{00}, x_{10}, x_{01}, x_{11})$ forms a square.
This implies, the number of vertices that is
  unique neighbor to both $v_{01}$ and $v_{10}$,
  is less than or equal to the number of squares that contains
    one vertex in $v_{01}$ and one vertex in $v_{10}$.

The final step of bounding the number of squares relies on
  the bound for degrees developed in the previous subsection.
The detail is explained in the proof.

\begin{proof}[Proof of Lemma~\ref{lem:small-set-LTC}]
  Given $c_1 = (v_{01}, v_{10})$.
    When a vertex $x_{11}$ is 
      uniquely connected to $v_{01}$
      and not neighbor to $v_{10}$
      or uniquely connected to $v_{10}$
      and not neighbor to $v_{01}$,
    $x_{11} \in c_0$.
  For those vertices that are neighbor to both $v_{10}$ and $v_{01}$,
    say $x_{10} \in v_{10}, x_{01} \in v_{01}$
    and $(x_{10}, x_{11}) \in E_{1*}, (x_{01}, x_{11}) \in E_{*1}$ are neighbors,
    one can find another vertex $x_{00} \in V_{00}$
    where $(x_{00}, x_{01}, x_{10}, x_{11}) \in F$.
  
  So
  \begin{equation}
    |c_0| \ge |N^{\textnormal{unique}}_{11}(v_{01})| + |N^{\textnormal{unique}}_{11}(v_{10})|
      - 2 |f|,
  \end{equation}
  where $f = \{(x_{00}, x_{10}, x_{01}, x_{11}) \in F
    : x_{00} \in V_{00}, x_{10} \in v_{10}, x_{01} \in v_{01}, x_{11} \in V_{11}\}$
    are the squares with vertices in $v_{10}$ and $v_{01}$.
  
  By Corollary~\ref{cor:freely-implies-lossless},
    $(V_{01}, V_{11}, E_{*1})$ is a $(c_\downarrow |V_{0*}|/|V_{01}|, \epsilon_\downarrow)$-lossless expander
    and $(V_{10}, V_{11}, E_{1*})$ is a $(c_\rightarrow |V_{*0}|/|V_{10}|, \epsilon_\rightarrow)$-lossless expander.
  Because $|v_{01}| < c_\downarrow |V_{0*}|$ and $|v_{10}| < c_\rightarrow |V_{*0}|$,
  by Lemma~\ref{lem:unique-expander}
  \begin{equation}
    |N^{\textnormal{unique}}_{11}(v_{01})| \ge (1-2\epsilon_\downarrow) w_\downarrow |v_{01}|,
  \end{equation}

  \begin{equation}
    |N^{\textnormal{unique}}_{11}(v_{10})| \ge (1-2\epsilon_\rightarrow) w_\rightarrow |v_{10}|.
  \end{equation}
  Now, we suffice to upper bound the number of squares $|f|$.

  We write the total number of squares as a sum 
    of the number of squares over each $x_{00} \in V_{00}$,
  \begin{equation}
    |f| = \sum_{x_{00} \in V_{00}} d_\rightarrow(x_{00}) d_\downarrow(x_{00}),
  \end{equation}
  where $d_\rightarrow(x_{00}) = \deg_{v_{01}}(x_{00}), 
    d_\downarrow(x_{00}) = \deg_{v_{10}}(x_{00})$.

  By locally minimal,
  \begin{equation} \label{eq:local-minimal}
    d_\rightarrow(x_{00})/w_\rightarrow + d_\downarrow(x_{00})/w_\downarrow 
    \le 1.
  \end{equation}

  Because $|v_{10}| < c_\downarrow|V_{0*}|/w_\uparrow$,
  by Lemma~\ref{lem:lossless-split},
    we can find 
    $d^1_{\downarrow}(x_{00}), d^2_{\downarrow}(x_{00})$
    such that
    
  \begin{equation} \label{eq:split-down}
    d_{\downarrow}(x_{00}) = d^1_{\downarrow}(x_{00})
    + d^2_{\downarrow}(x_{00}),
  \end{equation}

  \begin{equation} \label{eq:bound-d1-down}
    \sum_{x_{00} \in V_{00}} d^1_\downarrow(x_{00}) = |v_{10}|,
  \end{equation}
  
  \begin{equation} \label{eq:bound-d2-down}
    \{d^2_\downarrow(x_{00})\}_{x_{00} \in V_{00}}
    \preceq \{\epsilon_\downarrow w_\downarrow^
    {\lceil \frac{w_\uparrow}{\epsilon_\downarrow w_\downarrow} |v_{10}|\rceil}\},
  \end{equation}

  Similarly,
    because $|v_{01}| < c_\rightarrow|V_{*0}|/w_\leftarrow$,
    we can find 
    $d^1_{\rightarrow}(x_{00}), d^2_{\rightarrow}(x_{00})$
    such that

  \begin{equation} \label{eq:split-right}
    d_{\rightarrow}(x_{00}) = d^1_{\rightarrow}(x_{00})
    + d^2_{\rightarrow}(x_{00}),
  \end{equation}

  \begin{equation} \label{eq:bound-d1-right}
    \sum_{x_{00} \in V_{00}} d^1_\rightarrow(x_{00}) = |v_{01}|,
  \end{equation}

  \begin{equation} \label{eq:bound-d2-right}
    \{d^2_\rightarrow(x_{00})\}_{x_{00} \in V_{00}}
    \preceq \{\epsilon_\rightarrow w_\rightarrow^
    {\lceil \frac{w_\leftarrow}{\epsilon_\rightarrow w_\rightarrow} |v_{01}|\rceil}\}.
  \end{equation}

  Now,
  \begin{align*}
    |f| &= \sum_{x_{00} \in V_{00}} d_\rightarrow(x_{00}) d_\downarrow(x_{00}) \\
        &= \sum_{x_{00} \in V_{00}, i=1, j=1}^{i=2, j=2} d^i_\rightarrow(x_{00}) d^j_\downarrow(x_{00}),
  \end{align*}
  and we will bound 4 different cross terms individually.

  For $i=1, j=1$, 
    by Equation~\ref{eq:local-minimal}, Equation~\ref{eq:bound-d1-down}, Equation~\ref{eq:bound-d1-right},
  \begin{align*}
    &|v_{10}|/w_\downarrow + |v_{01}|/w_\rightarrow \\
    &= \sum_{x_{00} \in V_{00}} 
      d^1_\downarrow(x_{00})/w_\downarrow + d^1_\rightarrow(x_{00})/w_\rightarrow \\
    &\ge \sum_{x_{00} \in V_{00}}
      (d^1_\downarrow(x_{00})/w_\downarrow + d^1_\rightarrow(x_{00})/w_\rightarrow) \\
    & (d_\downarrow(x_{00})/w_\downarrow + d_\rightarrow(x_{00})/w_\rightarrow) \\
    &\ge \sum_{x_{00} \in V_{00}}
      (d^1_\downarrow(x_{00})/w_\downarrow + d^1_\rightarrow(x_{00})/w_\rightarrow)^2 \\
    &\ge \sum_{x_{00} \in V_{00}}
      4 (d^1_\downarrow(x_{00})/w_\downarrow) (d^1_\rightarrow(x_{00})/w_\rightarrow).
  \end{align*}

  Therefore,
  \begin{equation*}
    \sum_{x_{00}\in V_{00}} d^1_\downarrow(x_{00}) d^1_\rightarrow(x_{00})
    \le (w_\rightarrow|v_{10}| + w_\downarrow|v_{01}|)/4.
  \end{equation*}

  For the other cases where $i \ne 1$ or $j \ne 1$,
    we use the fact that when $A \preceq A'$ and $B \preceq B'$
    and all their elements are nonnegative,
    we have $\sum_{k=1}^{k_{\max}} a_k b_k \le \sum_{k=1}^{k_{\max}} a'_k b'_k$,
    where $\{a_k\}_{k=1}^{k_{\max}}, \{a'_k\}_{k=1}^{k_{\max}},\{b_k\}_{k=1}^{k_{\max}}, \{b'_k\}_{k=1}^{k_{\max}}$
      are the sorted sequences of $A, A', B, B'$ in the descending order
      with appended 0s.
  We show $\sum_{k=1}^{k_{\max}} a_k b_k \le \sum_{k=1}^{k_{\max}} a'_k b_k \le \sum_{k=1}^{k_{\max}} a'_k b'_k$.
  \begin{align}
    & \sum_{k=1}^{k_{\max}} a'_k b_k - \sum_{k=1}^{k_{\max}} a_k b_k \\
    &= \sum_{k=1}^{k_{\max}} (a'_k-a_k) b_k \\
    &= \sum_{k=1}^{k_{\max}} (\sum_{l=1}^k a'_l-a_l) (b_k-b_{k+1}) \\
    &\ge 0,
  \end{align}
  where $b_{k_{\max}+1}$ is defined as 0.
  The last inequality holds because each term is nonnegative.
    $\sum_{l=1}^k a'_l-a_l \ge 0$ because $A \preceq A'$
    and $b_k-b_{k+1} \ge 0$ for $1 \le k < k_{\max}$ because the sequence is sorted in the descending order,
    and $b_{k_{\max}}-b_{k_{\max}+1} = b_{k_{\max}} \ge 0$ because all elements are nonnegative.

  Therefore,
  \begin{equation*}
    \sum_{x_{00}\in V_{00}} d^1_\downarrow(x_{00}) d^2_\rightarrow(x_{00})
    \le \epsilon_\rightarrow w_\rightarrow |v_{10}|,
  \end{equation*}
  \begin{equation*}
    \sum_{x_{00}\in V_{00}} d^2_\downarrow(x_{00}) d^1_\rightarrow(x_{00})
    \le \epsilon_\downarrow w_\downarrow |v_{01}|,
  \end{equation*}
  \begin{align*}
    & \sum_{x_{00}\in V_{00}} d^2_\downarrow(x_{00}) d^2_\rightarrow(x_{00}) \\
    &\le \epsilon_\downarrow w_\downarrow \epsilon_\rightarrow w_\rightarrow 
    \min(\lceil \frac{w_\uparrow}{\epsilon_\downarrow w_\downarrow} |v_{10}| \rceil,
    \lceil \frac{w_\leftarrow}{\epsilon_\rightarrow w_\rightarrow} |v_{01}| \rceil),
  \end{align*}

  Overall,
  \begin{align*}
    & |f| \\
    \le& (w_\rightarrow |v_{10}| + w_\downarrow |v_{01}|)/4 \\
    +& \epsilon_\rightarrow w_\rightarrow |v_{10}|
     + \epsilon_\downarrow w_\downarrow |v_{01}| \\
    +& \epsilon_\downarrow w_\downarrow \epsilon_\rightarrow w_\rightarrow 
    \min(\lceil \frac{w_\uparrow}{\epsilon_\downarrow w_\downarrow} |v_{10}| \rceil,
    \lceil \frac{w_\leftarrow}{\epsilon_\rightarrow w_\rightarrow} |v_{01}| \rceil) \\
    \le& (w_\rightarrow |v_{10}| + w_\downarrow |v_{01}|)/4 \\
    +& \epsilon_\rightarrow (w_\rightarrow |v_{10}|)
    + \epsilon_\downarrow (w_\downarrow |v_{01}|) \\
    +& 2\min(w_\uparrow \epsilon_\rightarrow (w_\rightarrow |v_{10}|),
    w_\leftarrow \epsilon_\downarrow (w_\downarrow |v_{01}|)),
  \end{align*}
  where we use $\lceil \frac{w_\uparrow}{\epsilon_\downarrow w_\downarrow} |v_{10}| \rceil \le 2\frac{w_\uparrow}{\epsilon_\downarrow w_\downarrow} |v_{10}|$.
  If $|v_{10}| = 0$, the inequality holds trivially.
  Otherwise, because $w_\uparrow > w_\downarrow, \epsilon_\downarrow < 1, |v_{10}| \ge 1$,
  we have $\frac{w_\uparrow}{\epsilon_\downarrow w_\downarrow} |v_{10}| > 1$.
  Therefore,
  $\lceil \frac{w_\uparrow}{\epsilon_\downarrow w_\downarrow} |v_{10}| \rceil \le 2 \frac{w_\uparrow}{\epsilon_\downarrow w_\downarrow} |v_{10}|$.
  Similarly,
  $\lceil \frac{w_\leftarrow}{\epsilon_\rightarrow w_\rightarrow} |v_{01}| \rceil \le 2 \frac{w_\leftarrow}{\epsilon_\rightarrow w_\rightarrow} |v_{01}|$.
  
  By assumption $w_\uparrow \epsilon_\rightarrow \le \epsilon,
    \epsilon_\rightarrow \le \epsilon,
    \epsilon_\downarrow  \le \epsilon$,
    so $|f| \le (1+12 \epsilon)(w_\rightarrow |v_{10}| + w_\downarrow |v_{01}|)/4$.

  Therefore,
  \begin{align}
    &|c_0| \\
    &\ge |N^{\textnormal{unique}}_{11}(v_{01})| + |N^{\textnormal{unique}}_{11}(v_{10})| - 2 |f| \\
    &\ge ((1-2\epsilon_\rightarrow) w_\rightarrow |v_{10}| + (1-2\epsilon_\downarrow) w_\downarrow |v_{01}|) \\
    &- (1+12\epsilon)(w_\rightarrow |v_{10}| + w_\downarrow |v_{01}|)/2 \\
    &\ge (\frac{1}{2}-8\epsilon)(w_\rightarrow |v_{10}| + w_\downarrow |v_{01}|).
  \end{align}

\end{proof}

\begin{remark} \label{rem:skew}
It may seem unnatural that the graph is skewed,
  namely the condition in \ref{thm:1-sided-lossless-expander-with-symmetry}
  is not symmetric.
  This is originated from $\min(\lceil \frac{w_\uparrow}{\epsilon_\downarrow w_\downarrow} |v_{10}| \rceil,
  \lceil \frac{w_\leftarrow}{\epsilon_\rightarrow w_\rightarrow} |v_{01}| \rceil)$ used in the proof.
  This means our proof doesn't work when $w_\rightarrow = w_\downarrow$.
  We have attempted to find such a proof, however, we failed and maybe it is impossible.
    Because we didn't impose two-dimensional constraints similar to 
    product-expansion in \cite{panteleev2021quantum}
    and robust testability in \cite{dinur2021locally}.
  All we assumed here is that the parameters of $X^\updownarrow$ and $X^\leftrightarrow$ satisfy some inequalities.

Nevertheless, 
  if we additionally assume certain two-dimensional structures
  both properties can be achieved.
However, the proof is quite different and will not be discussed further here.

\end{remark}

\end{document}